\documentclass[11pt]{article}
\usepackage[authoryear,round]{natbib}
\usepackage{float}
\usepackage{graphicx,amsmath,amssymb,mathrsfs,amsfonts,amsthm,color,slashbox,amsbsy,ragged2e}
\usepackage{hyperref}
\usepackage{subfigure}
\usepackage{epsfig}
 \usepackage{lineno}

\newcommand{\Fb}{{\bar{F}}}
\newcommand{\xb}{{\boldsymbol{ x}}}

\newcommand{\al}{\alpha}
\newcommand{\be}{\beta}
\newcommand{\la}{\lambda}

\newcommand{\R}{\mathbb{R}}
\newcommand{\lab}{\boldsymbol{\la}}
\newcommand{\pb}{\boldsymbol{p}}
\newcommand{\ub}{\boldsymbol{u}}
\newcommand{\vb}{\boldsymbol{v}}

\newcommand{\mub}{\boldsymbol{\mu}}

\newcommand{\bq}{\begin{equation}}
\newcommand{\eq}{\end{equation}}
\newcommand{\bqs}{\begin{equation*}}
\newcommand{\eqs}{\end{equation*}}
\newcommand{\bqa}{\begin{eqnarray}}
\newcommand{\eqa}{\end{eqnarray}}
\newcommand{\bqas}{\begin{eqnarray*}}
\newcommand{\eqas}{\end{eqnarray*}}
\newcommand{\bc}{\begin{cases}}
\newcommand{\ec}{\end{cases}}
\newcommand{\bt}{\begin{thm}}
\newcommand{\et}{\end{thm}}

\newtheorem{theorem}{Theorem}[section]
\newtheorem{lemma}{Lemma}[section]
\newtheorem{example}{Example}[section]
\newtheorem{definition}{Definition}[section]

\title{Stochastic comparisons of the largest claim amounts from two sets of interdependent heterogeneous portfolios}
\author{ Hossein Nadeb, Hamzeh Torabi, Ali Dolati\\ \\
Department  of Statistics, Yazd University,  Yazd, Iran,\\
}

\begin{document}
\date{}
\maketitle
\begin{abstract}
Let $ X_{\la_1},\ldots,X_{\la_n}$ be dependent non-negative
random variables and
$Y_i=I_{p_i} X_{\la_i}$, $i=1,\ldots,n$, where $I_{p_1},\ldots,I_{p_n}$ are independent Bernoulli random variables independent of
$X_{\la_i}$'s, with ${\rm E}[I_{p_i}]=p_i$, $i=1,\ldots,n$. In
actuarial sciences, $Y_i$ corresponds to the claim amount in a
portfolio of risks. In this paper, we compare the largest claim amounts of two sets of
interdependent portfolios, in the
sense of usual stochastic order, when the variables in one set have the parameters
$\lambda_1,\ldots,\lambda_n$ and $p_1,\ldots,p_n$ and the variables in the other set  have
the parameters $\lambda^{*}_1,\ldots,\lambda^{*}_n$ and $p^*_1,\ldots,p^*_n$. For illustration,
we apply the results to some important models in actuary.
\end{abstract}


\section{Introduction}
Suppose that $X_{\la_i}$, with survival function $\Fb(x;\la_i)$, denotes the total random severities
of $i$th $(i=1,\ldots,n)$ policyholder in an insurance period, and let $I_{p_i}$ be a
Bernoulli random variable associated with $X_{\la_i}$, such that
$I_{p_i}=1$ whenever the $i$th policyholder makes random claim amounts
$X_{\la_i}$ and $I_{p_i}=0$ whenever does not make a claim. In this
notation, $Y_i= I_{p_i} X_{\la_i}$ is the claim amount associated with $i$th policyholder and $(Y_1,\ldots,Y_n)$ is said to be a portfolio of
risks. Further, consider another portfolio of risks $(Y^*_1,\ldots,Y^*_n)$ with the parameter vectors $\lab^*$ and $\pb^*$.

The annual premium is the amount paid by the policyholder as the cost
of the insurance cover being purchased. In fact, it is the primary
cost to the policyholder for assigning the risk to the insurer
which depends on the type of insurance. Determination of the
annual premium is one of the important problems in insurance
analysis. Deriving preferences
between random future gains or losses is an
appealing topic for the actuaries. For this purpose, stochastic orderings are very
helpful. Stochastic orderings have been extensively used in some
areas of sciences such as management science, financial economics,
insurance, actuarial science, operation research, reliability
theory, queuing theory and survival analysis. For more details on
stochastic orderings, we refer to M\"{u}ller and Stoyan
\cite{must}, Shaked and Shanthikumar \cite{ss} and Li and Li
\cite{lll}.

The problem of stochastic comparisons of some important statistics in $(Y_1,\ldots,Y_n)$ and $(Y^*_1,\ldots,Y^*_n)$, such as
the number of claims, $\sum_{i=1}^n I_{p_i}$, the aggregate claim amounts, $\sum_{i=1}^n Y_i$, the smallest, $Y_{1:n}=\min(Y_1,\ldots,Y_n)$,
and the largest claim amounts, $Y_{n:n}=\max(Y_1,\ldots,Y_n)$ in two portfolios, have been
discussed by many researchers in literature; see, e.g., Karlin and
Novikoff \cite{kar}, Ma \cite{ma}, Frostig \cite{fro}, Hu and Ruan
\cite{huru}, Denuit and Frostig \cite{defr}, Khaledi and Ahmadi
\cite{khah}, Zhang and Zhao \cite{zz}, Barmalzan et al.
\cite{bar1}, Li and Li \cite{lili}, Barmalzan et al.\cite{bar2018}, Barmalzan and Najafabadi
\cite{bana}, Barmalzan et al. \cite{bar3}, Barmalzan et al.
\cite{bar2}, Balakrishnan et al. \cite{baet} and Li and Li\cite{lili2}. 

When the critical situations occur, such as earthquakes, tornadoes and epidemics, the role of the insurance companies is very highlighted. Usually, in these situations many of policies are simultaneously at risk and the severities have a positive dependence. The most of published articles consider the case that the severities are independent, while sometimes this assumption is not satisfied.

Assume that $X_{\la_1},\ldots,X_{\la_n}$ are continuous and non-negative random variables with the joint distribution
function $H(x_1,\ldots,x_n)$, marginal distribution (survival) functions $F(x;\la_1),\ldots,F(x;\la_n)$ ($\Fb(x;\la_1),\ldots,\Fb(x;\la_n)$), and the copula $C$ through the relation $H(x_1,\ldots,x_n)=C\left(F(x;\la_1),\ldots,F(x;\la_n)\right)$ in the view of the Sklar's Theorem; see Nelsen \cite{nel}.

In this paper, we first focus on the stochastic comparisons of the largest claim amounts from two sets of heterogeneous portfolios in the sense of usual stochastic ordering, when the both portfolios include two policies. Then, some results in the case that the portfolios include more than two policies are provided.

The rest of the paper is organized as follows. In Section \ref{sec2}, we recall some definitions
and lemmas which will be used in the sequel. In Section \ref{sec3}, stochastic comparisons of the largest
claim amounts from two interdependent heterogeneous portfolios of risks in a general model in the sense
of the usual stochastic ordering is discussed. Also, some examples are illustrated to show the validity of the results.

\section{The basic definitions and some prerequisites}\label{sec2}
In this section, we recall some notions of stochastic orderings,
majorization, weakly majorization, copula and  some
useful lemmas which are helpful to  prove the main results.
Throughout the paper, we use the notations $ \Bbb R =
(-\infty,+\infty) $, $ \Bbb R_{+} = [0,+\infty) $ and $ \Bbb R_{++} = (0,+\infty) $
\begin{definition}
$ X $ is said to be smaller than $ Y $ in the usual stochastic ordering, denoted by $ X \leq_{\rm st} Y $, if $ \bar{F}(x)\leq\bar{G}(x) $ for all  $ x \in \Bbb R$.





\end{definition}

For a comprehensive discussion of various stochastic orderings, we
refer to Li and Li \cite{lll} and Shaked and Shanthikumar
\cite{ss}.

We also need the concept of majorization of vectors
and the Schur-convexity and Schur-concavity of functions. For a
comprehensive discussion of these topics, we refer to Marshall et
al. \cite{met}. We use the notation $ x_{1:n}\leq x_{2:n}\leq
\ldots \leq x_{n:n}$ to denote the increasing arrangement of the
components of the vector $ \boldsymbol{x} = (x_{1}, \ldots ,
x_{n})$.
\begin{definition}
{\rm The vector $ \boldsymbol{x} $ is said to be
\begin{itemize}
\item[ (i)] weakly submajorized by the vector $ \boldsymbol{y} $ (denoted by $ \boldsymbol{x}\preceq_{\rm w}\boldsymbol{y} $) if
$\sum_{i=j}^{n}x_{i:n}\leq \sum_{i=j}^{n}y_{i:n}$ for all $j = 1, \ldots , n $,

\item[ (ii)] weakly supermajorized by the vector $ \boldsymbol{y} $ (denoted by $ \boldsymbol{x}\mathop \preceq \limits^{{\mathop{\rm w}} }\boldsymbol{y} $) if $ \sum_{i=1}^{j}x_{i:n}\geq \sum_{i=1}^{j}y_{i:n} $ for all $ j = 1, \ldots , n $,

\item[ (iii)] majorized by the vector $ \boldsymbol{y} $ (denoted by $ \boldsymbol{x}\mathop \preceq \limits^{{\mathop{\rm m}} }\boldsymbol{y} $) if $ \sum_{i=1}^{n}x_{i}= \sum_{i=1}^{n}y_{i}$ and $\sum_{i=1}^{j}x_{i:n}\geq \sum_{i=1}^{j}y_{i:n}$ for all $j = 1, \ldots , n-1 $.
\end{itemize}}
\end{definition}
\begin{definition}
{\rm A real valued function $ \phi $ defined on a set $ \mathscr{A}\subseteq {\Bbb R}^{n} $ is said to be Schur-convex (Schur-concave) on $ \mathscr{A} $ if

\[
\boldsymbol{x} \mathop \preceq \limits^{{\mathop{\rm m}} }\boldsymbol{y} \quad \text{on}\quad \mathscr{A} \Longrightarrow \phi(\boldsymbol{x})\leq (\geq)\phi(\boldsymbol{y}).
\]}
\end{definition}


\begin{lemma}[Marshall et al.\cite{met}, Theorem 3.A.4]\label{3a4}
{\rm Let $ \mathscr{A}\subseteq \Bbb R $ be an open set and let $\phi:  \mathscr{A}^n\rightarrow  \Bbb R$ be continuously differentiable. $\phi$ is Schur-convex (Schur-concave) on $\mathscr{A}^n$ if and only if, $\phi$ is symmetric on  $\mathscr{A}^n$ and for all $i\neq j$,
\begin{equation*}
(x_i-x_j)\left(\frac{\partial{\phi(\boldsymbol{x})}}{\partial{x_i}}-\frac{\partial{\phi(\boldsymbol{x})}}{\partial{x_j}}\right)\geq(\leq)0, \quad \text{for all}\quad \boldsymbol{x}\in \mathscr{A}^n.
\end{equation*}
}
\end{lemma}


\begin{lemma}[Marshall et al.\cite{met}, Theorem 3.A.7]\label{mkl}
Let $\phi$ be a continuous real valued function on the set $\mathscr{D}=\{\xb:x_1\geq x_2\geq \ldots\geq x_n\}$ and continuously differentiable on the interior of $\mathscr{D}$. Denote the partial derivative of $\phi$ with respect to $i$th argument by $\phi_{(i)}(\boldsymbol{z})=\partial \phi(\boldsymbol{z})/\partial z_i$. Then,
\begin{equation*}
\phi(\xb)\leq \phi(\boldsymbol{y})\quad \text{whenever}~~  \boldsymbol{x}\preceq_{\rm w}\boldsymbol{y}~~  \text{on}~~ \mathscr{D}
\end{equation*}
if and only if
\begin{equation*}
\phi_{(1)}(\boldsymbol{z})\geq \phi_{(2)}(\boldsymbol{z})\geq \ldots\geq \phi_{(n)}(\boldsymbol{z})\geq 0,
\end{equation*}
i.e. the gradient $\bigtriangledown \phi(\boldsymbol{z})\in \mathscr{D}_{+}=\{\xb:x_1\geq x_2\geq \ldots\geq x_n\geq 0\}$, for all $\boldsymbol{z}$ in the interior of $\mathscr{D}$. Similarly,
\begin{equation*}
\phi(\xb)\leq \phi(\boldsymbol{y})\quad \text{whenever}~~  \boldsymbol{x}\mathop \preceq \limits^{{\mathop{\rm w}} }\boldsymbol{y}~~  \text{on}~~ \mathscr{D}
\end{equation*}
if and only if
\begin{equation*}
0\geq \phi_{(1)}(\boldsymbol{z})\geq \phi_{(2)}(\boldsymbol{z})\geq \ldots\geq \phi_{(n)}(\boldsymbol{z}),
\end{equation*}
i.e. the gradient $\bigtriangledown \phi(\boldsymbol{z})\in \mathscr{D}_{-}=\{\xb:0\geq x_1\geq x_2\geq \ldots\geq x_n\}$, for all $\boldsymbol{z}$ in the interior of $\mathscr{D}$.
\end{lemma}
One of the needed concepts in this paper is Archimedean copula. The class of Archimedean copula having a wide range of dependence structures including the independent copula. In the following, we state some useful definitions and lemmas related to copulas.

\begin{definition}
A copula $C$ is called Archimedean if is of the form $C(v_1,\ldots,v_n)=\phi^{-1}\left(\sum \limits_{i=1}^n \phi(v_i)\right)$, for $(v_1,\ldots,v_n)\in [0,1]^n$, which $\phi:[0,1]\rightarrow[0,\infty]$ is a strictly decreasing function, $\phi(0)=\infty$, $\phi(1)=0$ and $(-1)^k \frac{{\rm d}^k \phi (x)}{{\rm d}x^k}\geq 0 $, for $k\geq 0$ such that $\phi^{-1}$ is the inverse of $\phi$.
 The function $\phi$ is called generator of the copula.
\end{definition}

\begin{definition}
A two dimentional copula $C$ is positively quadrant dependent (PQD) if for all $(v_1,v_2)\in [0,1]^2$, we have $C(v_1,v_2)\geq v_1v_2$.
\end{definition}

\begin{definition}\label{def1}
Let $C_1$ and $C_2$ be two copulas. $C_1$ is less positively lower orthant dependent (PLOD) than $C_2$, denoted by $C_1\prec C_2$, if for all $\vb \in [0,1]^n$, $C_1(\vb)\leq C_2(\vb)$.
\end{definition}

We state the following lemmas from Durante\cite{dur} and Dolati and Dehghan Nezhad\cite{dd} related to Schur-concavity of copulas. 

\begin{lemma}\label{l6}
Let $C$ be a continuously differentiable copula. $C$ is Schur-concave on $[0,1]^n$, if and only if,
\begin{itemize}
\item[{\rm (i)}] $C$ is symmetric;
\item[{\rm (ii)}] $\frac{\partial C(\vb)}{\partial v_1}\geq \frac{\partial C(\vb)}{\partial v_2}$ on the set $\{\vb \in [0,1]^n: v_1\leq\ldots \leq v_n\}$.
\end{itemize}
\end{lemma}

\begin{lemma}\label{l7}
Every Archimedean copula is Schur-concave.
\end{lemma}
An important copula in application, is the Farlie-Gumbel-Morgenstern (FGM) copula which introduced by Morgenstern\cite{morgen} with a trace back to Eyraud\cite{eyr} and discussed by Gumbel\cite{gum} and Farlie\cite{far}, of the form $C_{\theta}(\vb)=\prod \limits_{i=1}^n v_i+\theta \prod \limits_{i=1}^n v_i (1-v_i)$, where $\theta\in [-1,1]$.

\begin{lemma}\label{l8}
The FGM copula is Schur-concave for any $\theta \in [-1,1]$.
\end{lemma}
For a comprehensive discussion in the topic of copula and the different types of dependency, one may refer to Nelsen\cite{nel}.

Also, we define a required space as below:
\[
S = \bigg\lbrace (\boldsymbol{x},\boldsymbol{y}) = \begin{bmatrix}
x_{1}\quad x_{2} \\
y_{1}\quad y_{2}
\end{bmatrix} : ( x_{i}-x_{j})(y_{i}-y_{j})\leq0, \quad i,j = 1,2\bigg\}.\]

\section{Main results}\label{sec3}
In this section, we compare the largest
claim amounts from two interdependent heterogeneous portfolios of risks in the sense
of the usual stochastic ordering. Also, we present some examples to illustrate the validity of the results.

The following theorem provides a comparison between the largest
claim amounts in two heterogeneous portfolio of risks, in terms of $\pb$.
\begin{theorem}\label{t1}
Let $ X_{\la_1}$ and $X_{\la_2} $ be non-negative random variables with
$ X_{\la_i} \thicksim \Fb(x;\la_i)$, $i = 1, 2 $, and associated copula $C$. Further, suppose that $I_{p_1}, 
I_{p_2}$  ($I_{p^*_1},I_{p^*_2} $) is a set of independent Bernoulli random variables, independent of the $X_{\la_i}$'s, with ${\rm E}[I_{p_i}]=p_i$ (${\rm E}[I_{p^*_i}]=p^*_i$), $i=1,2$. Assume that the following conditions hold:
\begin{itemize}
\item[{\rm (i)}] $h:(0,1]\rightarrow I\subset\R_{++}$ is a differentiable and strictly increasing concave function, with the log-concave inverse;
\item[{\rm (ii)}] $\Fb(x;\la)$ is decreasing in $\la$ for any $x\in \R_+$;
\item[{\rm (iii)}] $C$ is PQD.
\end{itemize}
Then, for $(\lab,h(\pb))\in S$ and  $(\lab,h(\pb^*))\in S$, we have
\begin{eqnarray*}
(h(p^*_1), h(p^*_2))  \mathop \preceq \limits^{{\mathop{\rm m}} } (h(p_1), h(p_2)) \Longrightarrow Y^*_{2:2} \leq_{{\rm st}}Y_{2:2}.
\end{eqnarray*}
\end{theorem}
\begin{proof}
Without loss of generality, we suppose that $\la_1 \leq \la_2$. For $(\lab,h(\pb))\in S$ and  $(\lab,h(\pb^*))\in S$, we have $h(p_1)\geq h(p_2)$ and $h(p^*_1)\geq h(p^*_2)$. Let $h^{-1}$ be the inverse of the function $h$, $u_i=h(p_i)$ and $u^*_i=h(p^*_i)$, for $i=1,2$. It can be easily verified that the distribution function of $Y_{2:2}$ is given by
\begin{eqnarray}\label{dist1}
G_{Y_{2:2}}(x)&=&\prod \limits_{i=1}^2 \bigg(1-h^{-1}(u_i) \Fb(x;\la_i)\bigg)\nonumber\\
&&+h^{-1}(u_1) h^{-1}(u_2) \bigg[C\big(F(x;\la_1),F(x;\la_2)\big)-F(x;\la_1) F(x;\la_2)\bigg].
\end{eqnarray}
Let
\begin{equation*}\label{psi}
G_{Y_{2:2}}(x)=-\Psi_1(\ub)-\Psi_2(\ub),
\end{equation*}
where
\begin{equation*}
\Psi_1(\ub)=-\prod \limits_{i=1}^2 \bigg(1-h^{-1}(u_i) \Fb(x;\la_i)\bigg),
\end{equation*}
and
\begin{equation*}
\Psi_2(\ub)=-h^{-1}(u_1) h^{-1}(u_2) \bigg[C\big(F(x;\la_1),F(x;\la_2)\big)-F(x;\la_1) F(x;\la_2)\bigg].
\end{equation*}
The partial derivative of $\Psi_1(\ub)$ with respect to $u_{i} $ is given by
\begin{equation*}
\frac{\partial \Psi_1(\ub)}{\partial u_i}=-\frac{\Fb(x;\la_i) \frac{{\rm d} h^{-1}(u_i)}{{\rm d} u_i}}{1-h^{-1}(u_i)\Fb(x;\la_i)} \Psi_1(\ub)\geq 0.
\end{equation*}
Since $\Fb(x;\la)$ is decreasing in $\la$, by using the increasing and convexity properties of $h^{-1}(x)$ in $x\in \R_{+}$, for $\la_1 \leq \la_2$ and $u_1\geq u_2$, we have
\begin{equation}\label{eq1}
0\leq 1-h^{-1}(u_1)\Fb(x;\la_1)\leq 1-h^{-1}(u_2)\Fb(x;\la_2),
\end{equation}
and
\begin{equation}\label{eq2}
\Fb(x;\la_1)\frac{{\rm d} h^{-1}(u_1)}{{\rm d} u_1}\geq \Fb(x;\la_2)\frac{{\rm d} h^{-1}(u_2)}{{\rm d} u_2} \geq 0.
\end{equation}
Using \eqref{eq1} and \eqref{eq2}, we obtain
\begin{equation*}
\frac{\partial \Psi_1(\ub)}{\partial u_1}-\frac{\partial \Psi_1(\ub)}{\partial u_2}=-\bigg[\frac{\Fb(x;\la_1)\frac{{\rm d} h^{-1}(u_1)}{{\rm d} u_1}}{1-h^{-1}(u_1)\Fb(x;\la_1)}-\frac{\Fb(x;\la_2)\frac{{\rm d} h^{-1}(u_2)}{{\rm d} u_2}}{1-h^{-1}(u_2)\Fb(x;\la_2)}\bigg]\Psi_1(\ub)\geq 0.
\end{equation*}
Applying the Lemma \ref{mkl} and the assumption $(u^*_1,u^*_2)\mathop \preceq \limits^{{\mathop{\rm m}} }(u_1,u_2)$, imply that
\begin{equation}\label{eq10}
\Psi_1(\ub^*)\leq \Psi_1(\ub).
\end{equation}

Now, the partial derivative of $\Psi_2(\ub)$ with respect to $u_{i} $ is given by
\begin{equation*}
\frac{\partial \Psi_2(\ub)}{\partial u_i}=\frac{\frac{{\rm d}h^{-1}(u_i)}{{\rm d}u_i}}{h^{-1}(u_i)}\Psi_2(\ub)=\frac{{\rm d} \log h^{-1}(u_i)}{{\rm d} u_i}\Psi_2(\ub)\leq 0.
\end{equation*}
 Therefore, for $u_1\geq u_2$, we obtain
\begin{equation*}
\frac{\partial \Psi_2(\ub)}{\partial u_1}-\frac{\partial \Psi_2(\ub)}{\partial u_2}=\bigg[\frac{{\rm d} \log h^{-1}(u_1)}{{\rm d} u_1}-\frac{{\rm d} \log h^{-1}(u_2)}{{\rm d} u_2}\bigg]\Psi_2(\ub)\geq 0,
\end{equation*}
where the inequality follows from log-concavity of $h^{-1}$ and negativity of $\Psi_2(\ub)$ which is due to PQD property of $C$. Thus, applying Lemma \ref{mkl} and the assumption $(u^*_1,u^*_2)\mathop \preceq \limits^{{\mathop{\rm m}} }(u_1,u_2)$, imply that
\begin{equation}\label{eq11}
\Psi_2(\ub^*)\leq \Psi_2(\ub).
\end{equation}
By using \eqref{eq10} and \eqref{eq11}, the proof is completed.
\end{proof}

The following theorem provides a comparison between the largest
claim amounts in two heterogeneous portfolio of risks, in terms of $\lab$.
\begin{theorem}\label{t2}
Let $ X_{\la_1}$ and $X_{\la_2} $ ($ X_{\la^{*}_1}$ and  $X_{\la^{*}_2} $) be non-negative random variables with
$ X_{\la_i} \thicksim \Fb(x;\la_i)$ ($ X_{\la^{*}_i} \thicksim \Fb(x;\la^*_{i} )$), $i = 1, 2 $, and associated copula $C$. Further, suppose that $I_{p_1}, I_{p_2}$ is a set of independent Bernoulli random variables, independent of the $X_{\la_i}$'s
($X_{\la^*_i}$'s), with ${\rm E}[I_{p_i}]=p_i$, $i=1,2$. Assume that the following conditions hold:
\begin{itemize}
\item[{\rm (i)}] $h:[0,1]\rightarrow I\subset\R_{+}$ is a differentiable and strictly increasing function;
\item[{\rm (ii)}] $\Fb(x;\la)$ is decreasing and convex in $\la$ for any $x\in \R_+$;
\item[{\rm (iii)}] $\frac{\partial C(v_1,v_2)}{\partial v_1}\geq \frac{\partial C(v_1,v_2)}{\partial v_2}$, for all $0\leq v_1\leq v_2\leq 1$.
\end{itemize}
Then, for $(\lab,h(\pb))\in S$ and  $(\lab^*,h(\pb))\in S$, we have
\begin{eqnarray*}
(\la^*_1, \la^*_2)  \mathop \preceq \limits^{{\mathop{\rm w}} } (\la_1, \la_2) \Longrightarrow Y^*_{2:2} \leq_{{\rm st}}Y_{2:2}.
\end{eqnarray*}
\end{theorem}
\begin{proof}
Without loss of generality, we suppose that $\la_1 \leq \la_2$, $u_1\geq u_2$ and $u^*_1\geq u^*_2$. By some algebraic calculations in \eqref{dist1}, the distribution function of $Y_{2:2}$ can be rewritten as the following form:
\begin{eqnarray*}
G_{Y_{2:2}}(x)&=&(1-h^{-1}(u_1))(1-h^{-1}(u_2))+h^{-1}(u_1)h^{-1}(u_2)\nonumber\\
&&\times \bigg[C\big(F(x;\la_1),F(x;\la_2)\big)+\frac{1-h^{-1}(u_2)}{h^{-1}(u_2)}F(x;\la_1)+\frac{1-h^{-1}(u_1)}{h^{-1}(u_1)} F(x;\la_2)\bigg].
\end{eqnarray*}
Define $\Psi(\lab)=-G_{Y_{2:2}}(x)$. The partial derivative of $\Psi(\lab)$ with respect to $\la_{i} $, $i=1,2$ are given by
\begin{equation*}
\frac{\partial \Psi(\lab)}{\partial \la_1}=-h^{-1}(u_1)h^{-1}(u_2)\frac{{\rm d}F(x;\la_1)}{{\rm d}\la_1}\bigg[\frac{\partial C\big(F(x;\la_1),F(x;\la_2)\big)}{\partial v_1}+\frac{1-h^{-1}(u_2)}{h^{-1}(u_2)}\bigg]\leq 0,
\end{equation*}
and
\begin{equation*}
\frac{\partial \Psi(\lab)}{\partial \la_2}=-h^{-1}(u_1)h^{-1}(u_2)\frac{{\rm d}F(x;\la_2)}{{\rm d}\la_2}\bigg[\frac{\partial C\big(F(x;\la_1),F(x;\la_2)\big)}{\partial v_2}+\frac{1-h^{-1}(u_1)}{h^{-1}(u_1)}\bigg]\leq 0,
\end{equation*}
where the inequalities are due to decreasing property of $\Fb(x;\la)$ in $\la$ and positivity of $\frac{1-h^{-1}(x)}{h^{-1}(x)}$ in $x\in \R_{+}$. Since $h^{-1}$ is increasing in $x\in \R_+$ and $\Fb(x;\la)$ is decreasing and convex in $\la$ for any $x\in \R_{+}$, then for $\la_1 \leq \la_2$ and $u_1\geq u_2$, we have
\begin{equation}\label{eq5}
0\leq \frac{1-h^{-1}(u_1)}{h^{-1}(u_1)}\leq \frac{1-h^{-1}(u_2)}{h^{-1}(u_2)},
\end{equation}
and
\begin{equation}\label{eq6}
\frac{{\rm d} F(x;\la_1)}{{\rm d}\la_1}\geq \frac{{\rm d} F(x;\la_2)}{{\rm d}\la_2} \geq 0.
\end{equation}
The decreasing property of $\Fb(x;\la)$ in $\la$ and the condition (iii) imply that
\begin{equation}\label{eq7}
\frac{\partial C\big(F(x;\la_1),F(x;\la_2)\big)}{\partial v_1}\geq \frac{\partial C\big(F(x;\la_1),F(x;\la_2)\big)}{\partial v_2}\geq 0.
\end{equation}
Using \eqref{eq5}, \eqref{eq6} and \eqref{eq7}, we obtain
\begin{eqnarray*}
\frac{\partial \Psi(\lab)}{\partial \la_2}-\frac{\partial \Psi(\lab)}{\partial \la_1}&=&-h^{-1}(u_1)h^{-1}(u_2)\\
&&\times\bigg[\frac{{\rm d}F(x;\la_2)}{{\rm d}\la_2}\frac{\partial C\big(F(x;\la_1),F(x;\la_2)\big)}{\partial v_2}+\frac{{\rm d}F(x;\la_2)}{{\rm d}\la_2}\frac{1-h^{-1}(u_1)}{h^{-1}(u_1)}\\
&&~~~-\frac{{\rm d}F(x;\la_1)}{{\rm d}\la_1}\frac{\partial C\big(F(x;\la_1),F(x;\la_2)\big)}{\partial v_1}-\frac{{\rm d}F(x;\la_1)}{{\rm d}\la_1}\frac{1-h^{-1}(u_2)}{h^{-1}(u_2)}\bigg]\geq 0.
\end{eqnarray*}
Therefore, under the assumption $\lab^*\mathop \preceq \limits^{{\mathop{\rm w}} } \lab$, Lemma \ref{mkl} implies that
\begin{equation*}
\Psi(\lab^*)\leq \Psi(\lab),
\end{equation*}
which completes the proof.
\end{proof}
The following theorem provides a comparison between the largest claim amounts in two heterogeneous portfolio of risks, in terms of $\pb$ and $\lab$.
\begin{theorem}\label{t3}
Let $ X_{\la_1}$ and $X_{\la_2} $ ($ X_{\la^*_1}$ and $X_{\la^*_2} $) be non-negative random variables with
$ X_{\la_i} \thicksim \Fb(x;\la_i)$ ($ X_{\la^*_i} \thicksim \Fb(x;\la^*_i)$), $i = 1, 2 $, and associated copula C. Further, suppose that $I_{p_1}, I_{p_2}$  ($I_{p^*_1},I_{p^*_2} $) is a set of independent Bernoulli random variables, independent of the $X_{\la_i}$'s ($X_{\la^*_i}$'s), with ${\rm E}[I_{p_i}]=p_i$ (${\rm E}[I_{p^*_i}]=p^*_i$), $i=1,2$. Assume that the following conditions hold:
\begin{itemize}
\item[{\rm (i)}] $h:(0,1]\rightarrow I\subset\R_{++}$ is a differentiable and strictly increasing concave function, with a log-concave inverse;
\item[{\rm (ii)}] $\Fb(x;\la)$ is decreasing and convex in $\la$ for any $x\in \R_+$;
\item[{\rm (iii)}] $C$ is PQD and $\frac{\partial C(v_1,v_2)}{\partial v_1}\geq \frac{\partial C(v_1,v_2)}{\partial v_2}$, for all $0\leq v_1\leq v_2\leq 1$.
\end{itemize}
Then, for $(\lab,h(\pb))\in S$ and  $(\lab^*,h(\pb^*))\in S$, we have
\begin{eqnarray*}
(h(p^*_1), h(p^*_2))  \mathop \preceq \limits^{{\mathop{\rm m}} } (h(p_1), h(p_2))~~\text{and}~~(\la^*_1,\la^*_2)\mathop \preceq \limits^{{\mathop{\rm w}} } (\la_1, \la_2) \Longrightarrow Y^*_{2:2} \leq_{{\rm st}}Y_{2:2}.
\end{eqnarray*}
\end{theorem}
\begin{proof}
Let $V_{2:2}$, $Z_{2:2}$ and $W_{2:2}$ be the largest claim amounts from the portfolios $(I_{p^*_{1:2}} X_{\la^*_{2:2}},I_{p^*_{2:2}} X_{\la^*_{1:2}})$, $(I_{p_{1:2}} X_{\la^*_{2:2}},I_{p_{2:2}}X_{\la^*_{1:2}})$ and $(I_{p_{1:2}} X_{\la_{2:2}},I_{p_{2:2}}X_{\la_{1:2}})$, respectively. It can be verified that $Y^*_{2:2}\mathop = \limits^{{\mathop{\rm st}} }V_{2:2}$ and $Y_{2:2}\mathop = \limits^{{\mathop{\rm st}} }W_{2:2}$. On the other hand, Theorem \ref{t1} and Theorem \ref{t2} imply that $V_{2:2} \leq_{{\rm st}}Z_{2:2}$ and $Z_{2:2} \leq_{{\rm st}}W_{2:2}$, respectively. Hence, the required result is obtained.
\end{proof}

The scale family is an applicable model in reliability theory and actuarial science. $X_{\la}$  is said to follow the scale family, if its survival function can be expressed as $\Fb(x;\la)=\Fb(\la x)$, where $\Fb(x)$ is the baseline survival function with the corresponding density function $f(x)$ and $\la>0$.

The following theorem provides a comparison between the largest
claim amounts in two heterogeneous portfolio of risks, whenever the marginal distributions belonging to the scale family.
\begin{theorem}\label{t5}
Let $\Fb(x;\la_i)=\Fb(\la_i x)$ and $\Fb(x;\la^*_i)=\Fb(\la^*_i x)$, for $i=1,2$. Under the setup of Theorem \ref{t3}, suppose that the following conditions hold:
\begin{itemize}
\item[{\rm (i)}] $h:(0,1]\rightarrow I\subset\R_{++}$ is a differentiable and strictly increasing concave function, with a log-concave inverse;
\item[{\rm (ii)}] $f(x)$ is decreasing in $x\in \R_+$;
\item[{\rm (iii)}] $C$ is PQD and $\frac{\partial C(v_1,v_2)}{\partial v_1}\geq \frac{\partial C(v_1,v_2)}{\partial v_2}$, for all $0\leq v_1\leq v_2\leq 1$.
\end{itemize}
Then, for $(\lab,h(\pb))\in S$ and  $(\lab^*,h(\pb^*))\in S$, we have
\begin{eqnarray*}
(h(p^*_1), h(p^*_2))  \mathop \preceq \limits^{{\mathop{\rm m}} } (h(p_1), h(p_2))~~\text{and}~~(\la^*_1,\la^*_2)\mathop \preceq \limits^{{\mathop{\rm w}} } (\la_1, \la_2) \Longrightarrow Y^*_{2:2} \leq_{{\rm st}}Y_{2:2}.
\end{eqnarray*}
\end{theorem}
\begin{proof}
Note that the conditions (i) and (iii) are similar to the conditions (i) and (iii) of Theorem \ref{t3}. Also, it can be easily verified that  the condition (ii) of this theorem, satisfies the condition (ii) of Theorem \ref{t3}, which holds the desired result.
\end{proof}

Gamma distribution is one of the most applicable distributions to depict the claim amounts whenever the shape parameter is less than 1. $X$ has the gamma distribution with the shape parameter $\al$ and the scale parameter $\la$, denoted by $X \thicksim \Gamma(\al,\la)$, if its density function is given by
\begin{equation*}
f(x;\al,\la)=\frac{\la^{\al}}{\Gamma(\al)}x^{\al-1} e^{-\la x},\quad x\in \R_{++}.
\end{equation*}
The following example provides a numerical example to illustrate the validity of Theorem \ref{t5}.
\begin{example}\label{ex1}
Let $ X_{\la_i}\thicksim \Gamma(0.8,\la_i) $ ($  X_{\la^*_i}\thicksim \Gamma(0.8,\la^*_i) $), for $i = 1, 2 $, with the associated FGM copula. It is clear that this copula is PQD if $\theta\in [0,1]$. Further, suppose that $I_{p_1}, I_{p_2}$  ($I_{p^*_1},I_{p^*_2} $) is a set of independent Bernoulli random variables, independent of the $X_{\la_i}$'s ($X_{\la^*_i}$'s), with ${\rm E}[I_{p_i}]=p_i$ (${\rm E}[I_{p^*_i}]=p^*_i$), for $i=1,2$. We take $h(p)=p$, $(\la_1,\la_2)=(0.26,0.74)$, $(p_1,p_2)=(0.03,0.02)$, $(\la^*_1,\la^*_2)=(0.4,0.6)$, $(p^*_1,p^*_2)=(0.026,0.024)$ and $\theta=0.5$. Using Lemma \ref{l6} and Lemma \ref{l8}, we get the condition (iii) of Theorem \ref{t5}, and obviously can be verified that the other conditions are also satisfied. So, we have $Y^*_{2:2}\leq_{{\rm st}}Y_{2:2}$. Figure \ref{fig1} represents the survival function of $Y_{2:2}$ and $Y^*_{2:2}$, which agrees with the intended result.
\end{example}

\begin{figure}[!h]
\centerline{\includegraphics[width=8cm,height=8cm]{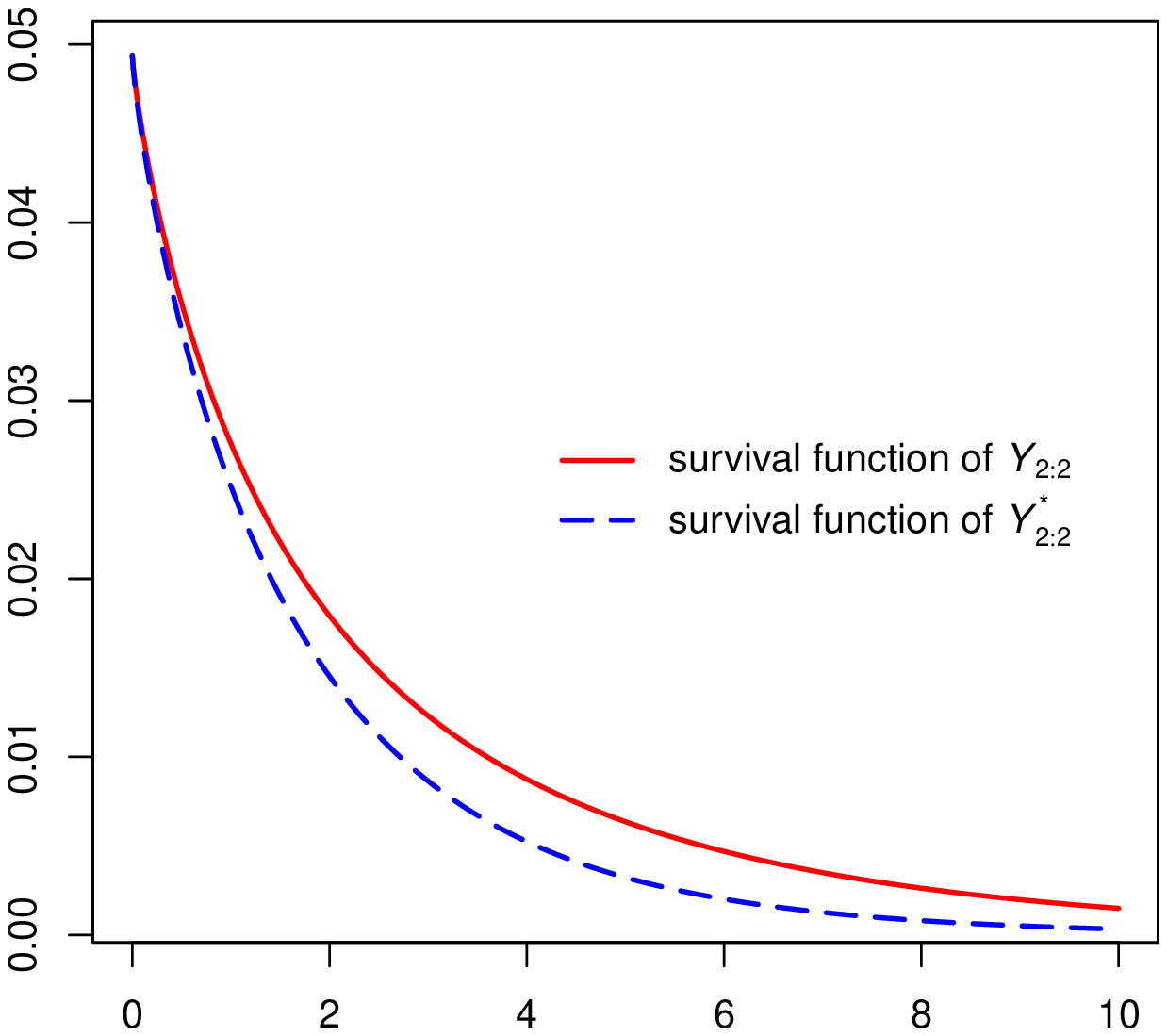}}
\vspace{-0.8cm} \caption{\small{ Plots of the survival functions of $Y_{2:2}$ and $Y^*_{2:2}$ in Example \ref{ex1}.}}\label{fig1}
\end{figure}

The following example illustrates that the conditions $(\lab,h(\pb))\in S$ and $(\lab^*,h(\pb^*))\in S$ is an important condition and can not be dropped.

\begin{example}\label{ex2}
Under the same setup in Example \ref{ex1}, we take $(p_1,p_2)=(0.02,0.03)$ and $(p^*_1,p^*_2)=(0.028,0.022)$ with the other unchanged values. It is clear that $(\lab,h(\pb))\notin S$, but it can be easily verified that the other conditions of Theorem \ref{t5} are satisfied. Figure \ref{fig2} represents the survival function of $Y_{2:2}$ and $Y^*_{2:2}$, which cross each other.
\end{example}

\begin{figure}[!h]
\centerline{\includegraphics[width=8cm,height=8cm]{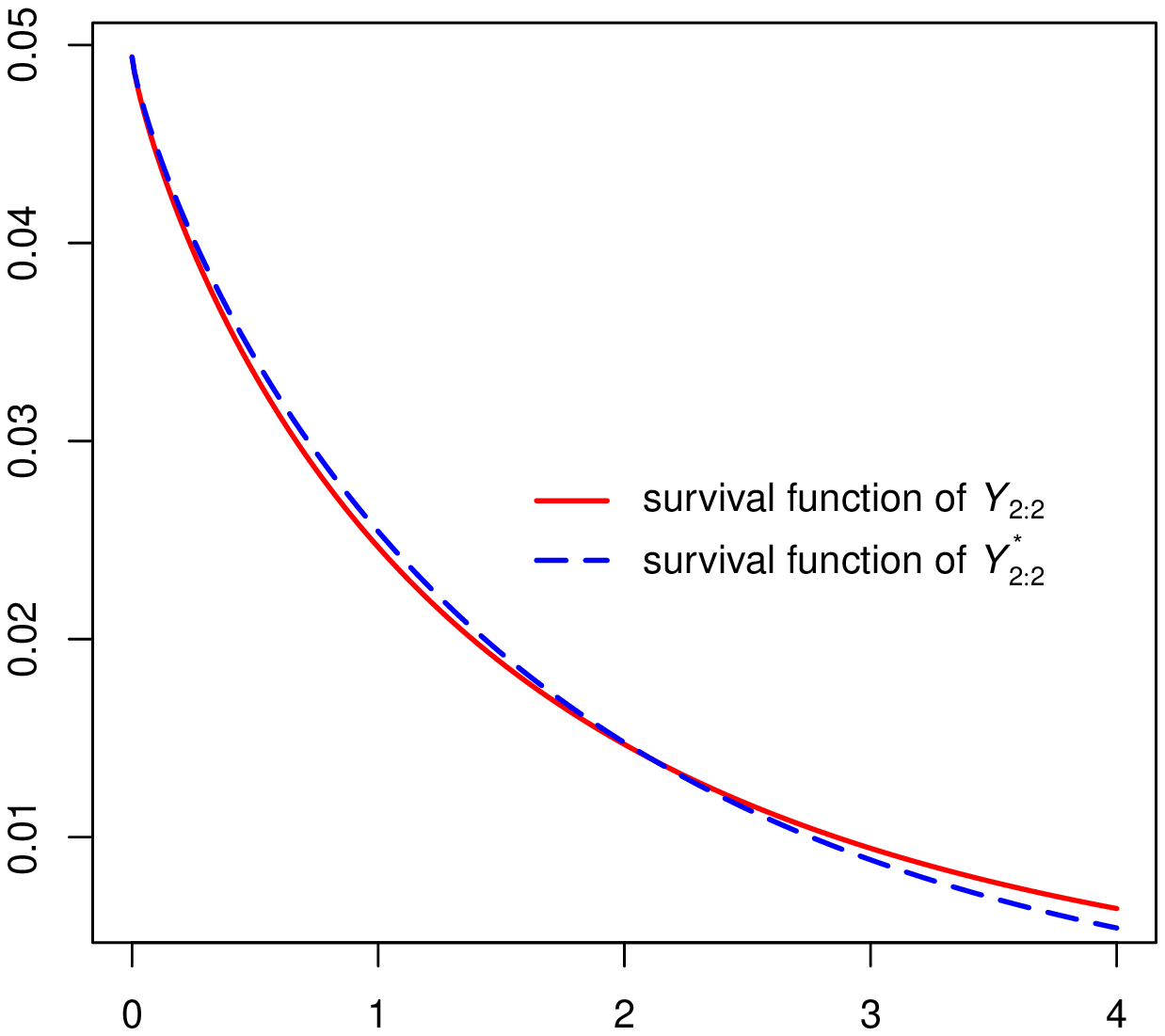}}
\vspace{-0.8cm} \caption{\small{ Plots of the survival functions of $Y_{2:2}$ and $Y^*_{2:2}$ in Example \ref{ex2}.}}\label{fig2}
\end{figure}

The proportional hazard rate (PHR) model is a flexible family of distributions with an important role in reliability theory, actuarial science and other fields; see for example Cox\cite{cox}, Finkelstein\cite{fin}, Kumar and Klefsj\"{o}\cite{kukl}, Balakrishnan et al.\cite{baet} and Li and Li\cite{lili2}.  $X_{\la}$ is said to follow PHR model, if its survival function can be expressed as $\Fb(x;\la)=[\Fb(x)]^{\la}$, where $\Fb(x)$ is the baseline survival function and $\la>0$.

The following theorem provides a comparison between the largest
claim amounts in two heterogeneous portfolio of risks, whenever the marginal distributions belonging to the PHR model.

\begin{theorem}\label{t6}
Let $\Fb(x;\la_i)=[\Fb(x)]^{\la_i}$ and $\Fb(x;\la^*_i)=[\Fb(x)]^{\la^*_i}$, for $i=1,2$. Under the setup of Theorem \ref{t3}, suppose that the following conditions hold:
\begin{itemize}
\item[{\rm (i)}] $h:(0,1]\rightarrow I\subset\R_{++}$ is a differentiable and strictly increasing concave function, with the log-concave inverse;
\item[{\rm (ii)}] $C$ is PQD and $\frac{\partial C(v_1,v_2)}{\partial v_1}\geq \frac{\partial C(v_1,v_2)}{\partial v_2}$, for all $0\leq v_1\leq v_2\leq 1$.
\end{itemize}
Then, for $(\lab,h(\pb))\in S$ and  $(\lab^*,h(\pb^*))\in S$, we have
\begin{eqnarray*}
(h(p^*_1), h(p^*_2))  \mathop \preceq \limits^{{\mathop{\rm m}} } (h(p_1), h(p_2))~~\text{and}~~(\la^*_1,\la^*_2)\mathop \preceq \limits^{{\mathop{\rm w}} } (\la_1, \la_2) \Longrightarrow Y^*_{2:2} \leq_{{\rm st}}Y_{2:2}.
\end{eqnarray*}
\end{theorem}
\begin{proof}
Note that $\Fb(x;\la)=[\Fb(x)]^{\la}$ is decreasing and convex in $\la$, which satisfies the condition (ii) of Theorem \ref{t3}. Therefore, applying Theorem \ref{t3} completes the proof.
\end{proof}

The Pareto distribution is a special case of the PHR model, which is commonly used as the distribution of claim severity from policyholders in insurance. $X$ has the Pareto distribution with parameters $\be$ and $\la$, denoted by $X \thicksim {\rm Pareto}(\beta,\la)$, if its survival function is given by
\begin{equation*}
\Fb(x;\be,\la)=(\frac{\beta}{x})^{\la},\quad x\geq \beta.
\end{equation*}
The following example provides a numerical example to illustrate the validity of Theorem \ref{t6}.
\begin{example}\label{ex3}
Let $ X_{\la_i}\thicksim {\rm Pareto}(1,\la_i) $ ($  X_{\la^*_i}\thicksim {\rm Pareto}(1,\la^*_i) $), for $i = 1, 2 $, with the associated Ali-Mikhail-Haq copula, which introduced by Ali et al.\cite{ali}, of the form $C_{\theta}(v_1,v_2)=\frac{v_1v_2}{1-\theta(1-v_1)(1-v_2)}$, where $\theta\in [-1,1]$. According to Nelsen\cite{nel}, this copula is Archimedean and obviously is PQD if $\theta\in [0,1]$ . Further, suppose that $I_{p_1}, I_{p_2}$  ($I_{p^*_1},I_{p^*_2} $) is a set of independent Bernoulli random variables, independent of the $X_{\la_i}$'s ($X_{\la^*_i}$'s), with ${\rm E}[I_{p_i}]=p_i$ (${\rm E}[I_{p^*_i}]=p^*_i$), for $i=1,2$. We take $h(p)=\log(p+2)$, $(\la_1,\la_2)=(4,2)$, $(p_1,p_2)=(0.02,0.06)$, $(\la^*_1,\la^*_2)=(4,6)$, $(p^*_1,p^*_2)=(0.0479,0.0319)$ and $\theta=0.3$. Lemma \ref{l6} and Lemma \ref{l7} imply the condition (ii) of Theorem \ref{t6}, and it can be easily verified that the other condition is also satisfied. So, we have $Y^*_{2:2}\leq_{{\rm st}}Y_{2:2}$. Figure \ref{fig3} represents the survival function of $Y_{2:2}$ and $Y^*_{2:2}$, which agrees with the intended result.
\end{example}

\begin{figure}[!h]
\centerline{\includegraphics[width=8cm,height=8cm]{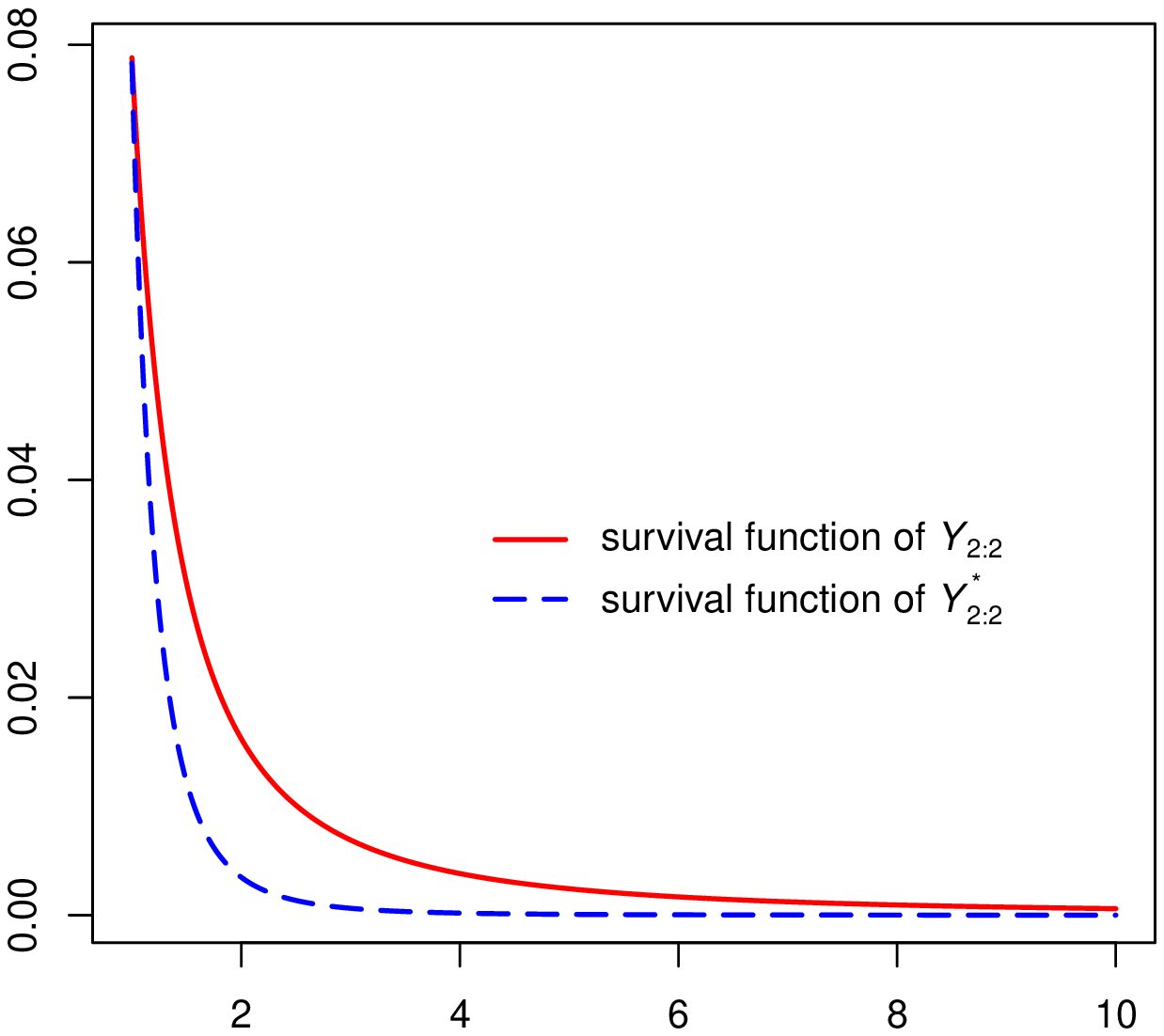}}
\vspace{-0.8cm} \caption{\small{ Plots of the survival functions of $Y_{2:2}$ and $Y^*_{2:2}$ in Example \ref{ex3}.}}\label{fig3}
\end{figure}

The transmuted-G $ ({\rm TG})$ model, which introduced by Mirhossaini and Dolati\cite{mir2} and Shaw and Buckley\cite{shbu}, is an attractive model for constructing new flexible distributions by adding a new parameter. The random variables $X_{\la}$ said to belong to the ${\rm TG}$ model with the baseline distribution function $F(x)$ and survival $\Fb(x)$, if its survival function can be expressed as $\Fb(x;\la)=\Fb(x)(1-\la F(x))$, where $ \la \in [-1,1]$.

The following theorem provides a comparison between the largest
claim amounts in two heterogeneous portfolio of risks, whenever the marginal distributions belonging to TG model.

\begin{theorem}\label{t7}
Let $\Fb(x;\la_i)=\Fb(x)(1-\la_i F(x))$ and $\Fb(x;\la^*_i)=\Fb(x)(1-\la^*_i F(x))$, for $i=1,2$. Under the setup of Theorem \ref{t3}, suppose that the following conditions hold:
\begin{itemize}
\item[{\rm (i)}] $h:(0,1]\rightarrow I\subset\R_{++}$ is a differentiable and strictly increasing concave function, with the log-concave inverse;
\item[{\rm (ii)}] $C$ is PQD and $\frac{\partial C(v_1,v_2)}{\partial v_1}\geq \frac{\partial C(v_1,v_2)}{\partial v_2}$, for all $0\leq v_1\leq v_2\leq 1$.
\end{itemize}
Then, for $(\lab,h(\pb))\in S$ and  $(\lab^*,h(\pb^*))\in S$, we have
\begin{eqnarray*}
(h(p^*_1), h(p^*_2))  \mathop \preceq \limits^{{\mathop{\rm m}} } (h(p_1), h(p_2))~~\text{and}~~(\la^*_1,\la^*_2)\mathop \preceq \limits^{{\mathop{\rm w}} } (\la_1, \la_2) \Longrightarrow Y^*_{2:2} \leq_{{\rm st}}Y_{2:2}.
\end{eqnarray*}
\end{theorem}
\begin{proof}
Note that $\Fb(x;\la)=\Fb(x)(1-\la F(x))$ is decreasing and convex in $\la$, which satisfies the condition (ii) of Theorem \ref{t3}. Therefore, applying Theorem \ref{t3} completes the proof.
\end{proof}

The transmuted exponential distribution, which introduced by Mirhossaini and Dolati\cite{mir2} has  non-negative support and can be used to simulate the claim severity from policyholders in insurance. $X$ has the transmuted exponential distribution with parameters $\mu$ and $\la$, denoted by $X \thicksim {\rm TE}(\mu,\la)$, if its survival function is given by
\begin{equation*}
\Fb(x,\mu,\la)=e^{-x/\mu}[1-\la(1-e^{-x/\mu})],\quad x\geq0,\quad \mu>0,\quad -1\leq \la \leq 1.
\end{equation*}

The following example provides a numerical example to illustrate the validity of Theorem \ref{t7}.
\begin{example}\label{ex4}
Let $ X_{\la_i}\thicksim {\rm TE}(3,\la_i) $ ($  X_{\la^*_i}\thicksim {\rm TE}(3,\la^*_i) $), for $i = 1, 2 $, with the associated Gumbel-Hougaard copula, which first introduced by Gumbel\cite{gum2}, of the form 
$$C_{\theta}(v_1,v_2)=\exp\bigg(-\left[(-\log v_1)^{\theta}+(-\log v_2)^{\theta}\right]^{1/\theta}\bigg),$$
 where $\theta\in [1,\infty)$. According to Nelsen\cite{nel}, this copula is Archimedean and is PQD. Further, suppose that $I_{p_1}, I_{p_2}$  ($I_{p^*_1},I_{p^*_2} $) is a set of independent Bernoulli random variables, independent of the $X_{\la_i}$'s ($X_{\la^*_i}$'s), with ${\rm E}[I_{p_i}]=p_i$ (${\rm E}[I_{p^*_i}]=p^*_i$), for $i=1,2$. We take $h(p)=\sqrt{p}$, $(\la_1,\la_2)=(0.6,-0.2)$, $(p_1,p_2)=(0.04,0.09)$, $(\la^*_1,\la^*_2)=(0.1,0.4)$, $(p^*_1,p^*_2)=(0.0676,0.0576)$ and $\theta=10$. Lemma \ref{l6} and Lemma \ref{l7} imply the condition (ii) of Theorem \ref{t7}, and it can be easily verified that the other condition is also satisfied. So, we have $Y^*_{2:2}\leq_{{\rm st}}Y_{2:2}$. Figure \ref{fig4} represents the survival function of $Y_{2:2}$ and $Y^*_{2:2}$, which coincides with the intended result.
\end{example}

\begin{figure}[!h]
\centerline{\includegraphics[width=8cm,height=8cm]{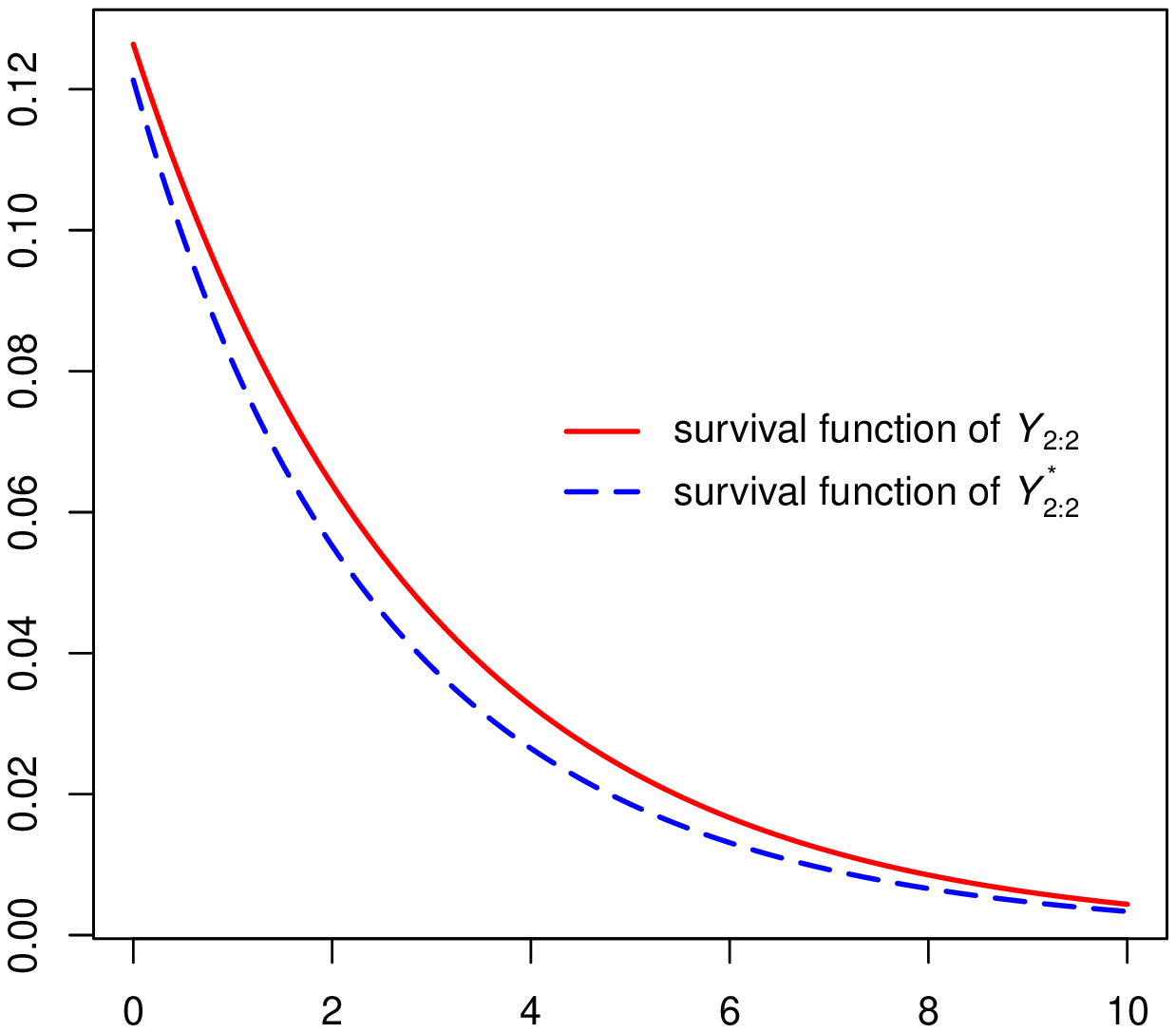}}
\vspace{-0.8cm} \caption{\small{ Plots of the survival functions of $Y_{2:2}$ and $Y^*_{2:2}$ in Example \ref{ex4}.}}\label{fig4}
\end{figure}

Next, we consider the case that the occurrence probabilities are also interdependent. Here, we denote $\boldsymbol{I}=(I_1,I_2) $ and $P(\boldsymbol{I}=\boldsymbol{\mu})=p(\boldsymbol{\mu})$. The following lemma considers the concept of weakly stochastic arrangement increasing through left tail probability ($\rm LWSAI$) for $\boldsymbol{I}$, which is a particular case of Lemma 5.3 of Cai and Wei\cite{cw}.

\begin{lemma}\label{l5}
A bivariate Bernoulli random vector $\boldsymbol{I}$ is $\rm LWSAI$, if and only if $p(1,0)\leq p(0,1)$.
\end{lemma}

The following theorem gives a comparison between the largest claim amounts in two heterogeneous portfolio of risks, whenever the occurrence probabilities are interdependent.

\begin{theorem}\label{t8}
Let $ X_{\la_1}$ and $X_{\la_2} $ ($ X_{\la^*_1}$ and $X_{\la^*_2} $) be non-negative random variables with
$ X_{\la_i} \thicksim \Fb(x;\la_i)$ ($ X_{\la^*_i} \thicksim \Fb(x;\la^*_i)$), $i = 1, 2 $, and associated copula C. Further, suppose that $\boldsymbol{I}$ is $\rm LWSAI$, and independent of the $X_{\la_i}$'s ($X_{\la^*_i}$'s). Assume that the following conditions hold:
\begin{itemize}
\item[{\rm (i)}] $\Fb(x;\la)$ is decreasing and convex in $\la$ for any $x\in \R_+$;
\item[{\rm (ii)}] $(\la^*_1,\la^*_2)\mathop \preceq \limits^{{\mathop{\rm m}} } (\la_1, \la_2)$, such that $\la_1\geq \la_2$ and $\la^*_1\geq \la^*_2$;
\item[{\rm (iii)}] $C$ is Schur-concave.
\end{itemize}
Then, we have $Y^*_{2:2} \leq_{{\rm st}}Y_{2:2}$.
\end{theorem}

\begin{proof}
Let $X_{2:2}=\max(X_{\la_1},X_{\la_2})$ and $X^*_{2:2}=\max(X_{\la^*_1},X_{\la^*_2})$. First, we prove that $X^*_{2:2}\leq_{{\rm st}} X_{2:2}$. It is enough to show that the function
\begin{equation*}
F_{X_{2:2}}(x)=C(F(x;\la_1),F(x;\la_2)),
\end{equation*}
is Schur-concave in $\lab$. According to Marshal et al.\cite{met}, Page 91, Table 2,  Schur-concavity of $C$ and increasing and concavity properties of $F(x;\la)$ in $\la$, implies that $F_{X_{2:2}}(x)$ is increasing and Schur-concave in $\lab$. Thus, condition (ii) implies
\begin{equation}\label{eq8}
X^*_{2:2}\leq_{{\rm st}} X_{2:2}.
\end{equation}
Also, according to Marshal et al.\cite{met}, the convexity of $\Fb(x;\la_i)$ in $\la_i$, implies the Schur-convexity of $\Fb(x;\la_1)+\Fb(x;\la_2)$ in $\lab$. Thus, the condition (ii) implies that
\begin{equation}\label{eq9}
\Fb(x;\la^*_1)+\Fb(x;\la^*_2)\leq \Fb(x;\la_1)+\Fb(x;\la_2).
\end{equation}
Note that
\begin{equation*}
G_{Y_{2:2}}(x)=p(0,0)+p(1,1)F_{X_{2:2}}(x)+p(0,1)F(x;\la_2)+p(1,0)F(x;\la_1),
\end{equation*}
and similarly,
\begin{equation*}
G_{Y^*_{2:2}}(x)=p(0,0)+p(1,1)F_{X^*_{2:2}}(x)+p(0,1)F(x;\la^*_2)+p(1,0)F(x;\la^*_1).
\end{equation*}
Thus, we have
\begin{eqnarray*}
G_{Y_{2:2}}(x)-G_{Y^*_{2:2}}(x)&=&p(1,1)[F_{X_{2:2}}(x)-F_{X^*_{2:2}}(x)]+p(0,1)[F(x;\la_2)-F(x;\la^*_2)]\\
&&+p(1,0)[F(x;\la_1)-F(x;\la^*_1)]\\
&=&p(1,1)[\Fb_{X^*_{2:2}}(x)-\Fb_{X_{2:2}}(x)]+p(0,1)[\Fb(x;\la^*_2)-\Fb(x;\la_2)]\\
&&+p(1,0)[\Fb(x;\la^*_1)-\Fb(x;\la_1)]\\
&\leq &p(0,1)[\Fb(x;\la^*_2)-\Fb(x;\la_2)]+p(1,0)[\Fb(x;\la^*_1)-\Fb(x;\la_1)]\\
&\leq &p(0,1)[\Fb(x;\la^*_2)-\Fb(x;\la_2)]+p(0,1)[\Fb(x;\la^*_1)-\Fb(x;\la_1)]\\
&= &p(0,1)[\Fb(x;\la^*_1)+\Fb(x;\la^*_2)-\Fb(x;\la_1)-\Fb(x;\la_2)]\\
&\leq& 0,
\end{eqnarray*}
where the first inequality is due to \eqref{eq8}, the second inequality is according to Lemma \ref{l5} and the last inequality is based on \eqref{eq9}. Hence, it is proved that $G_{Y_{2:2}}(x)\leq G_{Y^*_{2:2}}(x)$ which completes the proof.
\end{proof}

In the following, three special cases of Theorem \ref{t8} with respect to the scale, PHR and TG models, are represented.

\begin{theorem}\label{t9}
Let $\Fb(x;\la_i)=\Fb(\la_i x)$ and $\Fb(x;\la^*_i)=\Fb(\la^*_i x)$, for $i=1,2$. Under the setup of Theorem \ref{t8}, suppose that the following conditions hold:
\begin{itemize}
\item[{\rm (i)}] $f(x)$ is decreasing in $x\in \R_+$;
\item[{\rm (ii)}] $(\la^*_1,\la^*_2)\mathop \preceq \limits^{{\mathop{\rm m}} } (\la_1, \la_2)$, such that $\la_1\geq \la_2$ and $\la^*_1\geq \la^*_2$;
\item[{\rm (iii)}] $C$ is Schur-concave.
\end{itemize}
Then, we have $Y^*_{2:2} \leq_{{\rm st}}Y_{2:2}$.
\end{theorem}

\begin{proof}
Obviously, the condition (i) of Theorem \ref{t9} implies the condition (i) of Theorem \ref{t8} which completes the proof.
\end{proof}

\begin{theorem}\label{t10}
Let $\Fb(x;\la_i)=[\Fb(x)]^{\la_i}$ and $\Fb(x;\la^*_i)=[\Fb(x)]^{\la^*_i}$, for $i=1,2$. Under the setup of Theorem \ref{t8}, suppose that the following conditions hold:
\begin{itemize}
\item[{\rm (i)}] $(\la^*_1,\la^*_2)\mathop \preceq \limits^{{\mathop{\rm m}} } (\la_1, \la_2)$, such that $\la_1\geq \la_2$ and $\la^*_1\geq \la^*_2$;
\item[{\rm (ii)}] $C$ is Schur-concave.
\end{itemize}
Then, we have $Y^*_{2:2} \leq_{{\rm st}}Y_{2:2}$.
\end{theorem}

\begin{proof}
Obviously, $\Fb(x;\la)=[\Fb(x)]^{\la}$ satisfies the condition (i) of Theorem \ref{t8} which completes the proof.
\end{proof}

\begin{theorem}\label{t11}
Let $\Fb(x;\la_i)=\Fb(x)(1-\la_i F(x))$ and $\Fb(x;\la^*_i)=\Fb(x)(1-\la^*_i F(x))$, for $i=1,2$. Under the setup of Theorem \ref{t8}, suppose that the following conditions hold:
\begin{itemize}
\item[{\rm (i)}] $(\la^*_1,\la^*_2)\mathop \preceq \limits^{{\mathop{\rm m}} } (\la_1, \la_2)$, such that $\la_1\geq \la_2$ and $\la^*_1\geq \la^*_2$;
\item[{\rm (ii)}] $C$ is Schur-concave.
\end{itemize}
Then, we have $Y^*_{2:2} \leq_{{\rm st}}Y_{2:2}$.
\end{theorem}

\begin{proof}
Obviously, $\Fb(x;\la)=\Fb(x)(1-\la F(x))$ satisfies the condition (i) of Theorem \ref{t8} which completes the proof.
\end{proof}

The following example provides a numerical example to illustrate the validity of Theorem \ref{t10}.
\begin{example}\label{ex5}
Let $ X_{\la_i}\thicksim {\rm Pareto}(1,\la_i) $ ($  X_{\la^*_i}\thicksim {\rm Pareto}(1,\la^*_i) $), for $i = 1, 2 $, with the associated  FGM copula, with $\theta=0.7$. Let $(\la_1,\la_2)=(7,2)$, $(\la^*_1,\la^*_2)=(5.5,3.5)$, $p(0,0)=0.89$, $p(0,1)=0.06$, $p(1,0)=0.04$ and $p(1,1)=0.01$. Using Lemma \ref{l8}, we get the condition (ii) of Theorem \ref{t10}, and obviously can be verified that the other conditions are also satisfied. So, we have $Y^*_{2:2}\leq_{{\rm st}}Y_{2:2}$. Figure \ref{fig5} represents the survival function of $Y_{2:2}$ and $Y^*_{2:2}$, which approves with the intended result.
\end{example}

\begin{figure}[!h]
\centerline{\includegraphics[width=8cm,height=8cm]{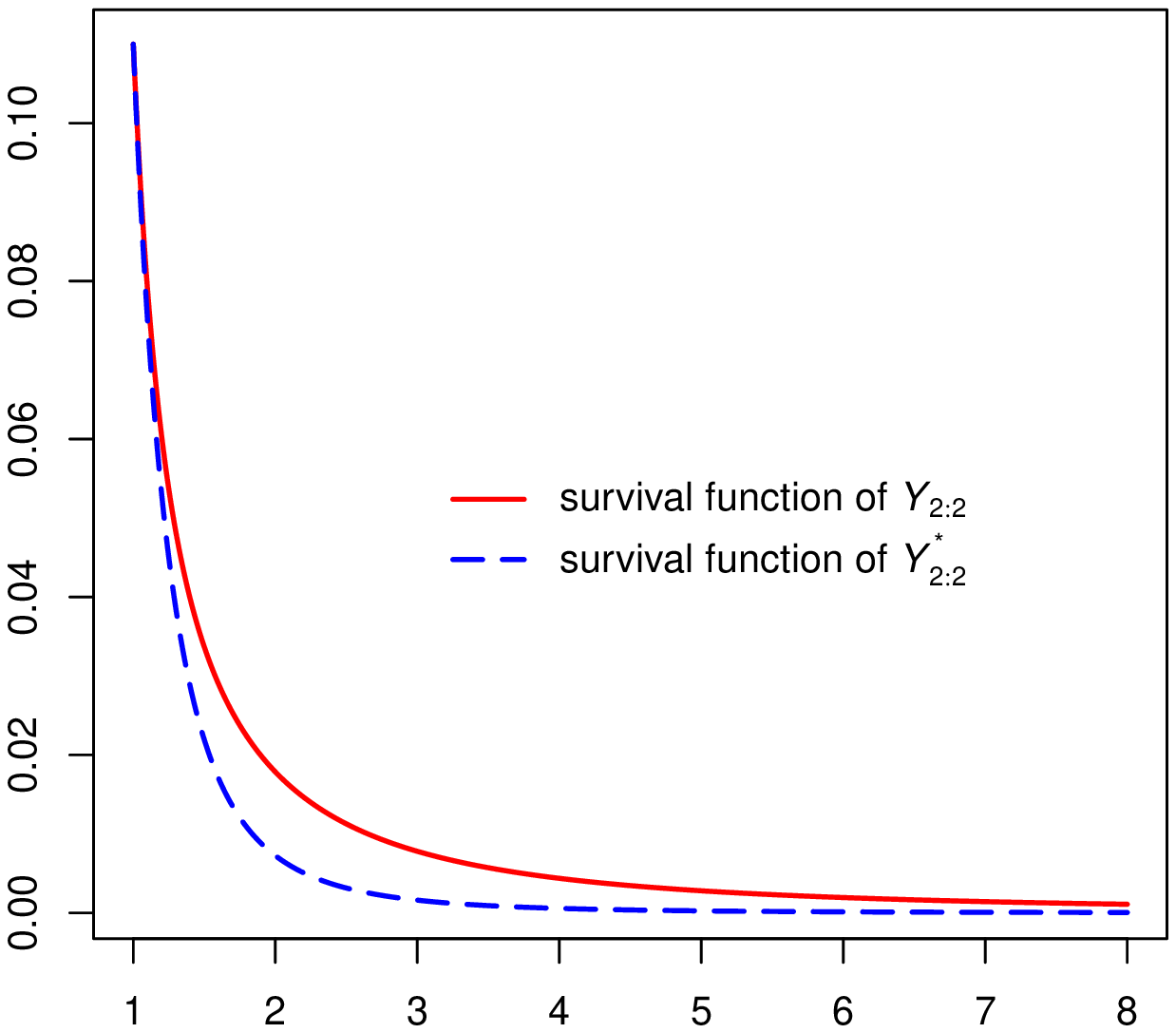}}
\vspace{-0.8cm} \caption{\small{ Plots of the survival functions of $Y_{2:2}$ and $Y^*_{2:2}$ in Example \ref{ex5}.}}\label{fig5}
\end{figure}

The following example illustrates that the conditions (ii) of Theorem \ref{t8} can not be dropped.

\begin{example}\label{ex6}
Under the same setup in Example \ref{ex5}, we take $(\la_1,\la_2)=(2,7)$ with the other unchanged values. It is clear that $\la_1 \ngeq \la_2$, but it can be easily verified that the other conditions of Theorem \ref{t8} are satisfied. Figure \ref{fig6} represents the survival function of $Y_{2:2}$ and $Y^*_{2:2}$, which cross each other.
\end{example}
\begin{figure}[!h]
\centerline{\includegraphics[width=8cm,height=8cm]{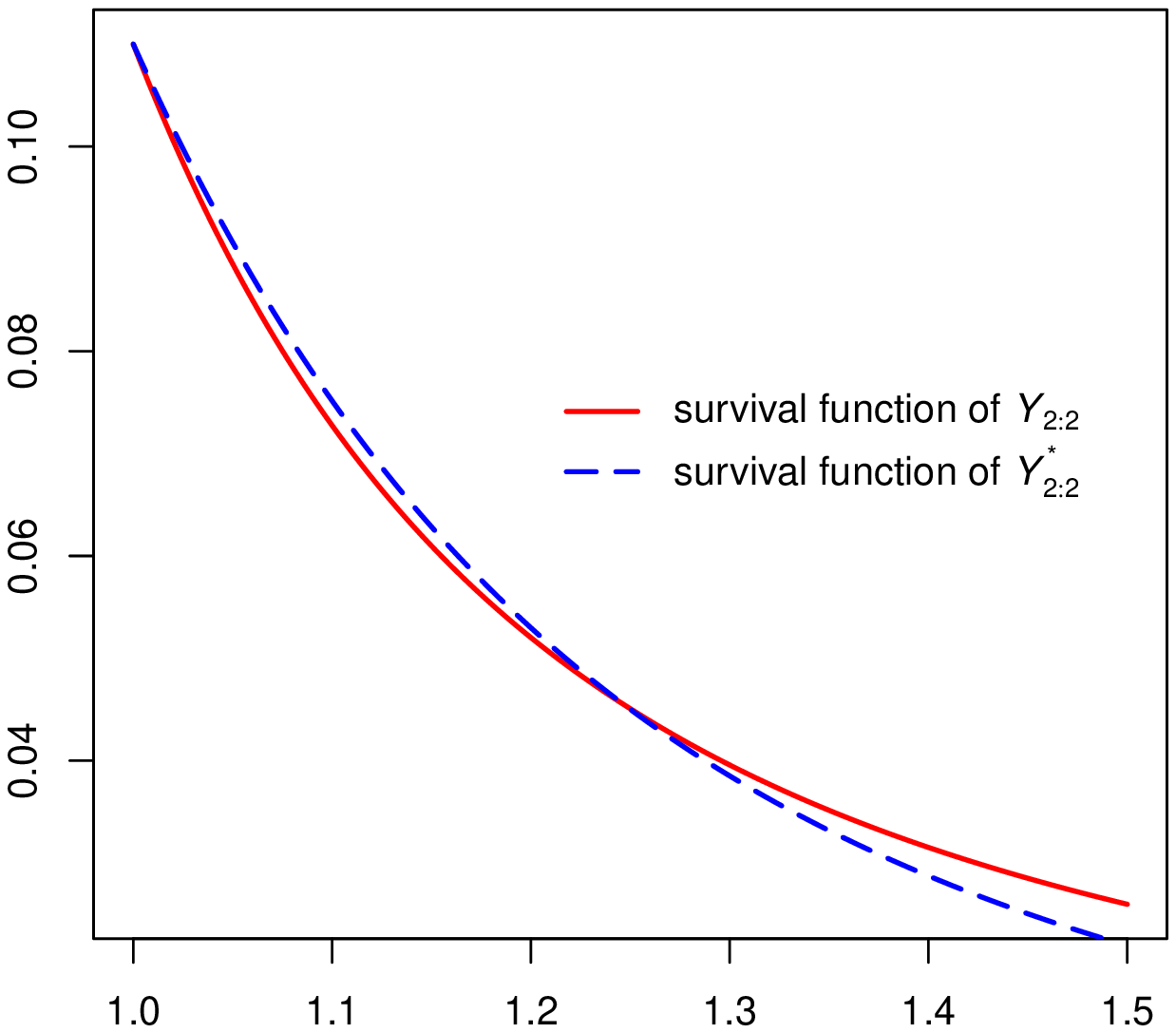}}
\vspace{-0.8cm} \caption{\small{ Plots of the survival functions of $Y_{2:2}$ and $Y^*_{2:2}$ in Example \ref{ex6}.}}\label{fig6}
\end{figure}

The following theorem provides a comparison between the largest claim amounts in two heterogeneous portfolios of risks, in terms of $\lab$.

\begin{theorem}\label{t12}
Let $ X_{\la_1},\ldots , X_{\la_n} $ ($ X_{\la^{*}_1}, \ldots, X_{\la^{*}_n} $) be non-negative random variables with
$ X_{\la_i} \thicksim \Fb(x;\la_i)$ ($ X_{\la^{*}_i} \thicksim \Fb(x;\la^*_{i} )$), $i = 1, \ldots,n $, and associated copula $C$. Further, suppose that $I_{p_1},\ldots, I_{p_n}$ is a set of independent Bernoulli random variables, independent of the $X_{\la_i}$'s
($X_{\la^*_i}$'s), with ${\rm E}[I_{p_i}]=p_i$, $i=1,\ldots,n$. Assume that $\Fb(x;\la)$ is decreasing in $\la$ for any $x\in \R_+$.
Then, we have
\begin{eqnarray*}
\la_i\leq \la^*_i,~ \forall~i=1,\ldots,n~\Longrightarrow Y^*_{n:n} \leq_{{\rm st}}Y_{n:n}.
\end{eqnarray*}
\end{theorem}
\begin{proof}
Denote $p(\mub)={\rm P}(I_{p_{1}}=\mu_1,\ldots,I_{p_{{n}}}=\mu_{n})$. The distribution function of $Y_{n:n}$ can be obtained as follows:
\begin{eqnarray}\label{eq13}
G_{Y_{n:n}}(x)&=&{\rm P}\bigg(Y_1\leq x,\ldots, Y_n\leq x\bigg)\nonumber\\
&=&{\rm P}\bigg(I_{p_{1}}X_{\la_{1}}\leq x,\ldots, I_{p_{{n}}}X_{\la_{{n}}}\leq x\bigg)\nonumber\\
&=& \sum \limits_{\mub \in \{0,1\}^{n}} p(\mub)~ {\rm P}\bigg(I_{p_{1}}X_{\la_{1}}\leq x,\ldots, I_{p_{{n}}}X_{\la_{{n}}}\leq x|I_{p_{1}}=\mu_1,\ldots, I_{p_{{n}}}=\mu_{n}\bigg)\nonumber\\
&=& \sum \limits_{\mub \in \{0,1\}^{n}} p(\mub)~ {\rm P}\bigg(\mu_1 X_{\la_{1}}\leq x,\ldots, \mu_n X_{\la_{{n}}}\leq x|I_{p_{1}}=\mu_1,\ldots, I_{p_{{n}}}=\mu_{n}\bigg)\nonumber\\
&=& \sum \limits_{\mub \in \{0,1\}^{n}} p(\mub)~ {\rm P}\bigg(\mu_1 X_{\la_{1}}\leq x,\ldots, \mu_n X_{\la_{{n}}}\leq x\bigg)\nonumber\\
&=&\sum \limits_{\mub \in \{0,1\}^{n}} p(\mub)~ C\bigg([F(x;\la_{1})]^{\mu_1},\ldots,[F(x;\la_{{n}})]^{\mu_{n}}\bigg).
\end{eqnarray} 
Based on decreasing property of $\Fb(x;\la)$ in $\la$ and the nature of copula, we immediately conclude that $G_{Y_{n:n}}(x)$ is increasing in $\la_i$, for $i=1,\ldots,n$. Hence, the desired result holds.
\end{proof}

The following theorem represents the impact due to degree of dependence in comparison the largest claim amounts in two heterogeneous portfolios of risks.

\begin{theorem}\label{t13}
Let $ X_{\la_1},\ldots , X_{\la_n} $ be non-negative random variables with
$ X_{\la_i} \thicksim \Fb(x;\la_i)$, $i = 1, \ldots,n $, and associated copula $C$ ($C^*$). In addition, suppose that $I_{p_1},\ldots, I_{p_n}$ is a set of independent Bernoulli random variables, independent of the $X_{\la_i}$'s, with ${\rm E}[I_{p_i}]=p_i$, $i=1,\ldots,n$. Then, we have
\begin{eqnarray*}
C\prec C^*~\Longrightarrow Y^*_{n:n} \leq_{{\rm st}}Y_{n:n}.
\end{eqnarray*}
\end{theorem}
\begin{proof}
By \eqref{eq13} and Definition \ref{def1}, the proof is immediately completed.
\end{proof}


The following theorem provides a comparison between the largest claim amounts in two heterogeneous portfolios of risks, in terms of $\lab$ and degree of dependence.

\begin{theorem}\label{t14}
Let $ X_{\la_1},\ldots , X_{\la_n} $ ($ X_{\la^{*}_1}, \ldots, X_{\la^{*}_n} $) be non-negative random variables with
$ X_{\la_i} \thicksim \Fb(x;\la_i)$ ($ X_{\la^{*}_i} \thicksim \Fb(x;\la^*_{i} )$), $i = 1, \ldots,n $, and associated copula $C$ ($C^*$). Furthermore, suppose that $I_{p_1},\ldots, I_{p_n}$ is a set of independent Bernoulli random variables, independent of the $X_{\la_i}$'s
($X_{\la^*_i}$'s), with ${\rm E}[I_{p_i}]=p_i$, $i=1,\ldots,n$. Assume that $\Fb(x;\la)$ is decreasing in $\la$ for any $x\in \R_+$.
Then, we have
\begin{eqnarray*}
C\prec C^* ~\text{and}~\la_i\leq \la^*_i,~ \forall~i=1,\ldots,n~\Longrightarrow Y^*_{n:n} \leq_{{\rm st}}Y_{n:n}.
\end{eqnarray*}
\end{theorem}
\begin{proof}
Let $V_{n:n}$, $Z_{n:n}$ and $W_{n:n}$ be the largest claim amounts from the portfolios $(I_{p_1} X_{\la^*_1},\ldots,I_{p_n}  X_{\la^*_n})$ with associated copula $C^*$, $(I_{p_1} X_{\la_1},\ldots,I_{p_n}X_{\la_n})$ with associated copula $C^*$, and $(I_{p_1} X_{\la_1},\ldots,I_{p_n}X_{\la_n})$ with associated copula $C$, respectively. It is easily seen that $Y^*_{n:n}\mathop = \limits^{{\mathop{\rm st}} }V_{n:n}$ and $Y_{n:n}\mathop = \limits^{{\mathop{\rm st}} }W_{n:n}$. On the other hand, Theorem \ref{t13} and Theorem \ref{t14} imply that $V_{n:n} \leq_{{\rm st}}Z_{n:n}$ and $Z_{n:n} \leq_{{\rm st}}W_{n:n}$, respectively. Hence, the proof is completed.
\end{proof}

The three following theorems consider the scale, PHR and TG models as the special cases of Theorem \ref{t14}.

\begin{theorem}\label{t15}
Let $\Fb(x;\la_i)=\Fb(\la_i x)$ and $\Fb(x;\la^*_i)=\Fb(\la^*_i x)$, for $i=1,\ldots,n$. Under the setup of Theorem \ref{t14},
Then, we have $Y^*_{n:n} \leq_{{\rm st}}Y_{n:n}$.
\end{theorem}

\begin{theorem}\label{t16}
Let $\Fb(x;\la_i)=[\Fb(x)]^{\la_i}$ and $\Fb(x;\la^*_i)=[\Fb(x)]^{\la^*_i}$, for $i=1,\ldots,n$. Under the setup of Theorem \ref{t14}, we have $Y^*_{n:n} \leq_{{\rm st}}Y_{n:n}$.
\end{theorem}

\begin{theorem}\label{t17}
Let $\Fb(x;\la_i)=\Fb(x)(1-\la_i F(x))$ and $\Fb(x;\la^*_i)=\Fb(x)(1-\la^*_i F(x))$, for $i=1,\ldots,n$. Under the setup of Theorem \ref{t14}, we have $Y^*_{n:n} \leq_{{\rm st}}Y_{n:n}$.
\end{theorem}

Another important distribution used as the distribution of claim severity from policyholders is Weibull distribution, which is a special case of the scale model. $X$ has the Weibull distribution with parameters $\al$ and $\la$, denoted by $X \thicksim {\rm Wei}(\al,\la)$, if its survival function is given by
\begin{equation*}
\Fb(x;\al,\la)=e^{-(\la x)^\al},\quad x> 0.
\end{equation*}
The following example provides a numerical example to illustrate the validity of Theorem \ref{t15}.

\begin{example}\label{ex7}
Let $ X_{\la_i}\thicksim {\rm Wei}(3,\la_i) $ ($  X_{\la^*_i}\thicksim {\rm Wei}(3,\la^*_i) $), for $i = 1, 2,3$, with the associated Frank copula, which introduced by Frank\cite{frank}, of the form $$C_{\theta}(v_1,v_2,v_3)=-\frac{1}{\theta}\log \left(1+\frac{(e^{-\theta v_1}-1)(e^{-\theta v_2}-1)(e^{-\theta v_3}-1)}{(e^{-\theta}-1)^2}\right),$$
 where $\theta\in(0,\infty)$. Further, suppose that $I_{p_1}, I_{p_2}, I_{p_3}$ is a set of independent Bernoulli random variables, independent of the $X_{\la_i}$'s ($X_{\la^*_i}$'s), with ${\rm E}[I_{p_i}]=p_i$ , for $i=1,2,3$. We take $(\la_1,\la_2,\la_3)=(0.5,0.7,0.3)$, $(\la^*_1,\la^*_2,\la^*_3)=(0.51,0.7,0.33)$, $(p_1,p_2,p_3)=(0.01,0.02,0.07)$ and $\theta=0.6$. Obviously, the conditions of Theorem \ref{t15} are satisfied. So, we have $Y^*_{3:3}\leq_{{\rm st}}Y_{3:3}$. Figure \ref{fig7} represents the survival function of $Y_{3:3}$ and $Y^*_{3:3}$, which coincides with the intended result.
\end{example}

\begin{figure}[!h]
\centerline{\includegraphics[width=8cm,height=8cm]{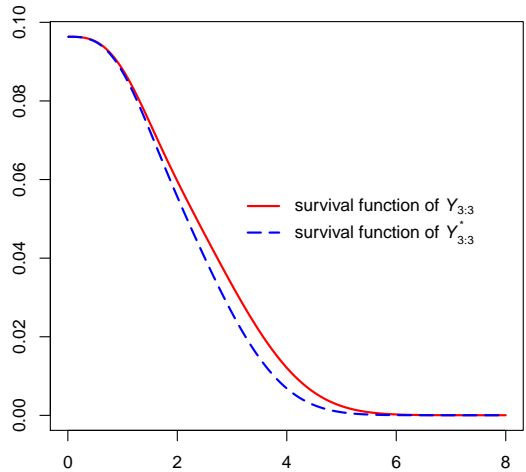}}
\vspace{-0.8cm} \caption{\small{ Plots of the survival functions of $Y_{3:3}$ and $Y^*_{3:3}$ in Example \ref{ex7}.}}\label{fig7}
\end{figure}

\section*{Conclusion}
In this paper, under some certain conditions, we discussed stochastic comparisons between the
largest claim amounts under dependency of severities in the sense of usual stochastic ordering in a general model, which particularly includes some important models such as the scale, PHR and TG models. However, we applied some distributions to illustrate the results.

\end{document}